\documentclass[11pt]{article}

\usepackage[a4paper,margin=2.4cm]{geometry}
\usepackage{amsmath,amssymb,bm}
\usepackage{amsthm}
\usepackage{graphicx}
\usepackage{booktabs}
\usepackage{array}
\usepackage{microtype}
\usepackage{authblk}
\usepackage{natbib}
\usepackage{hyperref}
\usepackage{siunitx}
\usepackage{xcolor}
\usepackage{enumitem}

\hypersetup{colorlinks=true,linkcolor=blue,citecolor=blue,urlcolor=blue}
\newtheorem{proposition}{Proposition}
\newtheorem{corollary}{Corollary}

\title{Fisher-Information Geometry Linking Schr\"odinger Dynamics, Bipartite Correlations, and Decoherence}
\author[1]{Juan Sumaya-Mart\'inez\thanks{ORCID: 0000-0002-7032-8824}}
\author[2]{Marcos C. Gonz\'alez\thanks{ORCID: 0000-0002-0933-6374}}
\affil[1]{Faculty of Sciences, Universidad Aut\'onoma del Estado de M\'exico, Toluca 50000, M\'exico}
\affil[2]{TecNM Tecnol\'ogico de Estudios Superiores de Jocotitl\'an, Jocotitl\'an, M\'exico, MX}
\date{}

\begin{document}
\maketitle

\begin{abstract}
Fisher information appears in quantum theory in two mathematically distinct settings: as a classical probability functional in variational reconstructions of Schr\"odinger dynamics and as the quantum Fisher metric of parameterized density operators. Here we examine how far these roles can be connected without identifying the two objects. Starting from a Hamilton--Jacobi ensemble action, we recover the Fisher term that generates the quantum potential and isolate its contribution to the kinetic-energy expectation. We then analyze the quantum Fisher information matrix (QFIM) of bipartite states. For any pure two-qubit state, the optimized magnitude of the off-diagonal QFIM element for normalized local generators equals the concurrence, while a collective phase generator yields $F_Q=4C^2$. Extending the optimized pure-state quantity by a convex roof gives an exact identity with Wootters concurrence for arbitrary mixed two-qubit states. This differs from evaluating the QFIM directly on a mixed density operator: for Werner states the raw optimized cross-QFIM is $2p^2/(1+p)$ and remains nonzero inside the separable region. The gap between the raw and convex-roof quantities therefore isolates, for this family, correlation geometry not captured by entanglement. Finally, in a two-path interferometer with a which-way marker, the optimal phase Fisher information available from local output statistics is $I_\phi^{\rm opt}=V^2$. These results place probability-gradient energy, bipartite correlation geometry, entanglement, and decoherence-induced loss of local phase sensitivity in a common mathematical setting while preserving the distinctions between classical Fisher information, QFIM geometry, Bell nonlocality, and global phase quantization.
\end{abstract}


\section{Introduction}
Quantum mechanics is exceptionally successful as a predictive theory, yet the physical status of some of its mathematical structures remains unsettled. The wave function evolves unitarily, whereas measurements are described in terms of probabilistic outcomes. Rather than attempting to settle the interpretation problem as a whole, we ask a narrower question: to what extent can dynamical and correlation structures of quantum mechanics be expressed within a common information-geometric framework?

There is a well-established route in this direction. Madelung showed that the Schr\"odinger equation can be written as coupled continuity and Hamilton--Jacobi-like equations for a density and a phase \citep{Madelung1927}; Bohm later emphasized the associated quantum potential in a realist formulation \citep{Bohm1952}. Reginatto approached the problem from the opposite direction and showed that adding a minimum-Fisher-information term to classical ensemble mechanics produces the nonrelativistic Schr\"odinger dynamics \citep{Reginatto1998,Reginatto1999}. Hall and Reginatto subsequently obtained the same structure from an exact uncertainty principle for nonclassical momentum fluctuations \citep{HallReginatto2002}. These results do not prove that Fisher information is an ontological substance, but they establish that it can sit inside the dynamical principle rather than merely being calculated after the fact.

A second line of work begins with statistical distinguishability. The quantum Fisher information introduced in quantum estimation theory supplies the local metric associated with parameter changes of a quantum state \citep{BraunsteinCaves1994,PetzGhinea2011}. It underlies quantum-enhanced metrology \citep{Giovannetti2006,Giovannetti2011,Pezze2018,TothApellaniz2014} and provides experimentally useful witnesses of multipartite entanglement \citep{Hyllus2012,Toth2012}. Quantum correlations, however, exhibit a hierarchy: separability defines the unentangled sector, whereas entanglement, steering, and Bell nonlocality are increasingly restrictive notions \citep{HorodeckiReview2009,Wiseman2007}. For this reason, any attempt to identify ``Fisher information'' directly with ``entanglement'' is too coarse.

The purpose of this paper is to place these two uses of Fisher geometry in the same framework without treating them as the same quantity. The spatial Fisher functional $I_F[\rho]$ is a classical functional of a configuration-space probability density, whereas the QFIM is the symmetric-logarithmic-derivative metric on a parameterized family of density operators. Their domains and operational meanings are different; the connection developed here is structural rather than an identification. Five results organize the analysis. First, the spatial Fisher functional generates the term required to pass from classical Hamilton--Jacobi ensemble dynamics to the local Madelung form of Schr\"odinger dynamics. Second, for pure two-qubit states an optimized cross component of the local QFIM equals concurrence once a fixed normalization of the local generators is imposed. Third, the convex-roof extension of that optimized pure-state quantity is exactly Wootters concurrence for every mixed two-qubit state. Fourth, direct evaluation of the mixed-state QFIM for Werner states gives a different, nonzero quantity in part of the separable region, thereby separating raw parameter-space correlation from convex-roof entanglement. Fifth, in a two-path interferometer the phase Fisher information available from local fringes decreases in lockstep with visibility as which-way information is transferred to a marker.

The claim is consequently limited but quantitative. We do not identify classical Fisher information with QFI, identify a raw QFIM entry with entanglement, claim that Fisher information solves the measurement problem, or attempt to restore a local hidden-variable description. Instead, we compare the roles played by information geometry in probability-gradient energy, parameter-space distinguishability, bipartite correlation, and two-qubit entanglement, and determine explicitly which identities are exact and which require an optimization or convex-roof construction.

\section{Schr\"odinger dynamics from Fisher information}
Consider a classical ensemble described by a normalized probability density $\rho(\mathbf{x},t)$ and Hamilton principal function $S(\mathbf{x},t)$. We assume that $\rho$ and $S$ are sufficiently smooth on each nodal domain, that the Fisher integral is finite, and that the variations either vanish at the boundary or decay sufficiently rapidly at infinity so that the integrations by parts below generate no surface terms. The classical ensemble action is
\begin{equation}
\mathcal{A}_{\rm cl}=\int dt\,d\mathbf{x}\,\rho\left[\partial_t S+\frac{|\nabla S|^2}{2m}+V(\mathbf{x},t)\right].
\end{equation}
Variation with respect to $S$ gives the continuity equation
\begin{equation}
\partial_t\rho+\nabla\cdot\left(\rho\frac{\nabla S}{m}\right)=0,
\end{equation}
whereas variation with respect to $\rho$ yields the classical Hamilton--Jacobi equation
\begin{equation}
\partial_tS+\frac{|\nabla S|^2}{2m}+V=0.
\end{equation}

The classical Fisher information defines a local statistical sensitivity and underlies the Cram\'er--Rao theory of estimation \citep{Fisher1925,Rao1945,Cramer1946,AmariNagaoka2000}. The spatial Fisher information used here is
\begin{equation}
I_F[\rho]=\int d\mathbf{x}\,\frac{|\nabla\rho|^2}{\rho}=4\int d\mathbf{x}\,|\nabla\sqrt{\rho}|^2.
\end{equation}
Following the information-theoretic construction, we consider
\begin{equation}
\mathcal{A}=\mathcal{A}_{\rm cl}+\frac{\hbar^2}{8m}\int dt\,I_F[\rho].
\end{equation}
Variation with respect to $\rho$ now gives
\begin{equation}
\partial_tS+\frac{|\nabla S|^2}{2m}+V-\frac{\hbar^2}{2m}\frac{\nabla^2\sqrt{\rho}}{\sqrt{\rho}}=0.
\end{equation}
The additional term
\begin{equation}
Q[\rho]=-\frac{\hbar^2}{2m}\frac{\nabla^2\sqrt{\rho}}{\sqrt{\rho}}
\end{equation}
has the form of the quantum potential. Introducing the Madelung transformation
\begin{equation}
\psi(\mathbf{x},t)=\sqrt{\rho(\mathbf{x},t)}\exp\left[\frac{iS(\mathbf{x},t)}{\hbar}\right],
\end{equation}
the continuity equation and modified Hamilton--Jacobi equation combine into
\begin{equation}
i\hbar\partial_t\psi=\left[-\frac{\hbar^2}{2m}\nabla^2+V\right]\psi.
\end{equation}
Thus, at the level of the ensemble variational formulation,
\begin{equation}
\boxed{\text{Hamilton--Jacobi dynamics}+\text{Fisher cost}=\text{Schr\"odinger dynamics}.}
\end{equation}
This equivalence is local and must be supplemented by the usual global single-valuedness or circulation condition on the phase. For a closed loop $\Gamma$ that does not cross a node of the wave function, the standard quantum sector satisfies
\begin{equation}
\oint_{\Gamma}\nabla S\cdot d\boldsymbol\ell=2\pi n\hbar,\qquad n\in\mathbb Z.
\end{equation}
Madelung variables can otherwise admit hydrodynamic solutions with nonquantized circulation and are therefore not globally equivalent to a single-valued Schr\"odinger wave function. This point, emphasized by Wallstrom, is a genuine limitation of hydrodynamic and stochastic reconstructions and is not removed merely by adding the Fisher functional \citep{Wallstrom1994}. In what follows we work within the sector satisfying this standard quantum phase quantization condition.

The kinetic-energy expectation can also be decomposed as
\begin{equation}
\frac{\hbar^2}{2m}\int |\nabla\psi|^2d\mathbf{x}
=\int\rho\frac{|\nabla S|^2}{2m}d\mathbf{x}+\frac{\hbar^2}{8m}I_F[\rho].
\label{eq:kinetic-decomp}
\end{equation}
Equation~\eqref{eq:kinetic-decomp} gives the Fisher term a direct energetic interpretation within the Madelung representation: it is the part of the kinetic expectation associated with gradients of the probability amplitude. Narrower spatial structures generally require steeper gradients and therefore a larger Fisher contribution. This statement is a property of the decomposition, not an additional postulate about the ontology of information.

\section{Local quantum Fisher geometry for bipartite systems}
Let a bipartite pure state undergo independent local unitary transformations,
\begin{equation}
|\psi(\theta_A,\theta_B)\rangle=e^{-i(\theta_A G_A+\theta_B G_B)}|\psi\rangle,
\end{equation}
where $G_A=g_A\otimes I$ and $G_B=I\otimes g_B$. For a pure state undergoing unitary parameter encoding, the QFIM reduces to the covariance form \citep{BraunsteinCaves1994,TothApellaniz2014,Pezze2018}
\begin{equation}
F_{\mu\nu}=4\,\mathrm{Re}\langle\Delta G_\mu\Delta G_\nu\rangle.
\end{equation}
Hence
\begin{equation}
\mathbf F=4\begin{pmatrix}
(\Delta G_A)^2 & \mathrm{Cov}(G_A,G_B)\\
\mathrm{Cov}(G_A,G_B) & (\Delta G_B)^2
\end{pmatrix}.
\end{equation}
The diagonal elements quantify local sensitivity, whereas the off-diagonal element
\begin{equation}
F_{AB}=4\,\mathrm{Cov}(G_A,G_B)
\end{equation}
measures the statistical coupling between the two local parameter directions. An individual off-diagonal QFIM component is not invariant under an arbitrary rescaling or reparameterization of those directions. To make the comparison below unambiguous, we therefore restrict the local qubit generators to $G_A=(\mathbf n\cdot\boldsymbol\sigma)/2$ and $G_B=(\mathbf m\cdot\boldsymbol\sigma)/2$, with $|\mathbf n|=|\mathbf m|=1$. This fixes their spectral range and removes a trivial scaling freedom. For a pure product state all connected correlations between strictly local observables vanish, so $F_{AB}=0$ for every pair of such local generators.

\section{Pure two-qubit states}
Any pure two-qubit state can be brought by local unitaries to Schmidt form,
\begin{equation}
|\psi_\alpha\rangle=\cos\alpha\,|00\rangle+\sin\alpha\,|11\rangle,
\qquad 0\le \alpha\le \frac{\pi}{4}.
\end{equation}
The concurrence is \citep{Wootters1998}
\begin{equation}
C(\alpha)=\sin(2\alpha).
\end{equation}
For local generators
\begin{equation}
G_A=\frac{\sigma_x^A}{2},\qquad G_B=\frac{\sigma_x^B}{2},
\end{equation}
one has $\langle\sigma_x^A\rangle=\langle\sigma_x^B\rangle=0$ and
\begin{equation}
\langle\sigma_x^A\sigma_x^B\rangle=\sin(2\alpha).
\end{equation}
Therefore
\begin{equation}
F_{AB}=\sin(2\alpha)=C.
\end{equation}

\subsection{Optimized cross-QFIM theorem for pure two-qubit states}
\begin{proposition}
For any pure two-qubit state $|\psi\rangle$, let the local generators be $G_A=(\mathbf n\cdot\boldsymbol\sigma)/2$ and $G_B=(\mathbf m\cdot\boldsymbol\sigma)/2$ with unit vectors $\mathbf n$ and $\mathbf m$. Then the maximum magnitude of the off-diagonal pure-state QFI matrix element equals the concurrence:
\begin{equation}
\max_{\mathbf n,\mathbf m}|F_{AB}|=C(|\psi\rangle).
\end{equation}
\end{proposition}
\begin{proof}
Local unitaries bring every pure two-qubit state to Schmidt form without changing concurrence. For $|\psi_\alpha\rangle=\cos\alpha|00\rangle+\sin\alpha|11\rangle$, the connected Pauli correlation tensor $T^{(c)}_{ij}=\langle\sigma_i\otimes\sigma_j\rangle-\langle\sigma_i\otimes I\rangle\langle I\otimes\sigma_j\rangle$ is diagonal with singular values $\{C,C,C^2\}$, where $C=\sin2\alpha$. For the normalized generators above, the cross-QFI element is $F_{AB}=\mathbf n^T T^{(c)}\mathbf m$. Maximization over unit vectors gives the operator norm of $T^{(c)}$, namely $C$ because $0\le C\le1$. Local-unitary invariance completes the proof.
\end{proof}

Define the optimized local cross-QFIM quantity
\begin{equation}
\mathcal F_\times(|\psi\rangle)=\max_{\mathbf n,\mathbf m}\left|F_{AB}\left(\frac{\mathbf n\cdot\boldsymbol\sigma}{2},\frac{\mathbf m\cdot\boldsymbol\sigma}{2}\right)\right|.
\end{equation}
For the Schmidt state, the connected correlation tensor has singular values $\{C,C,C^2\}$. Since $0\le C\le1$, the maximum singular value is $C$. Local-unitary invariance then gives
\begin{equation}
\boxed{\mathcal F_\times=C}
\end{equation}
for pure two-qubit states with the generator normalization fixed above. The equality should be read as a normalized tangent-space identity, not as a coordinate-free statement about arbitrary QFIM entries. It is closely related to covariance-matrix characterizations of entanglement \citep{Gittsovich2010}; the present contribution is to express the same connected-correlation structure directly as an optimized cross component of the local QFI metric and to follow its breakdown in mixed states.

A complementary result follows from the collective generator
\begin{equation}
J_z=\frac{1}{2}(\sigma_z^A+\sigma_z^B).
\end{equation}
For $|\psi_\alpha\rangle$,
\begin{equation}
\langle J_z\rangle=\cos(2\alpha),\qquad \langle J_z^2\rangle=1,
\end{equation}
so
\begin{equation}
F_Q[J_z]=4(\Delta J_z)^2=4\sin^2(2\alpha)=4C^2.
\end{equation}
Thus different parameter directions probe different aspects of the same state geometry.

An important negative result follows if the parameter is the Schmidt angle itself. Direct differentiation gives $F_Q^{(\alpha)}=4$ for all $\alpha$. Therefore an arbitrary QFI cannot be identified with entanglement; the generator matters.

\begin{figure}[t]
\centering
\includegraphics[width=0.88\linewidth]{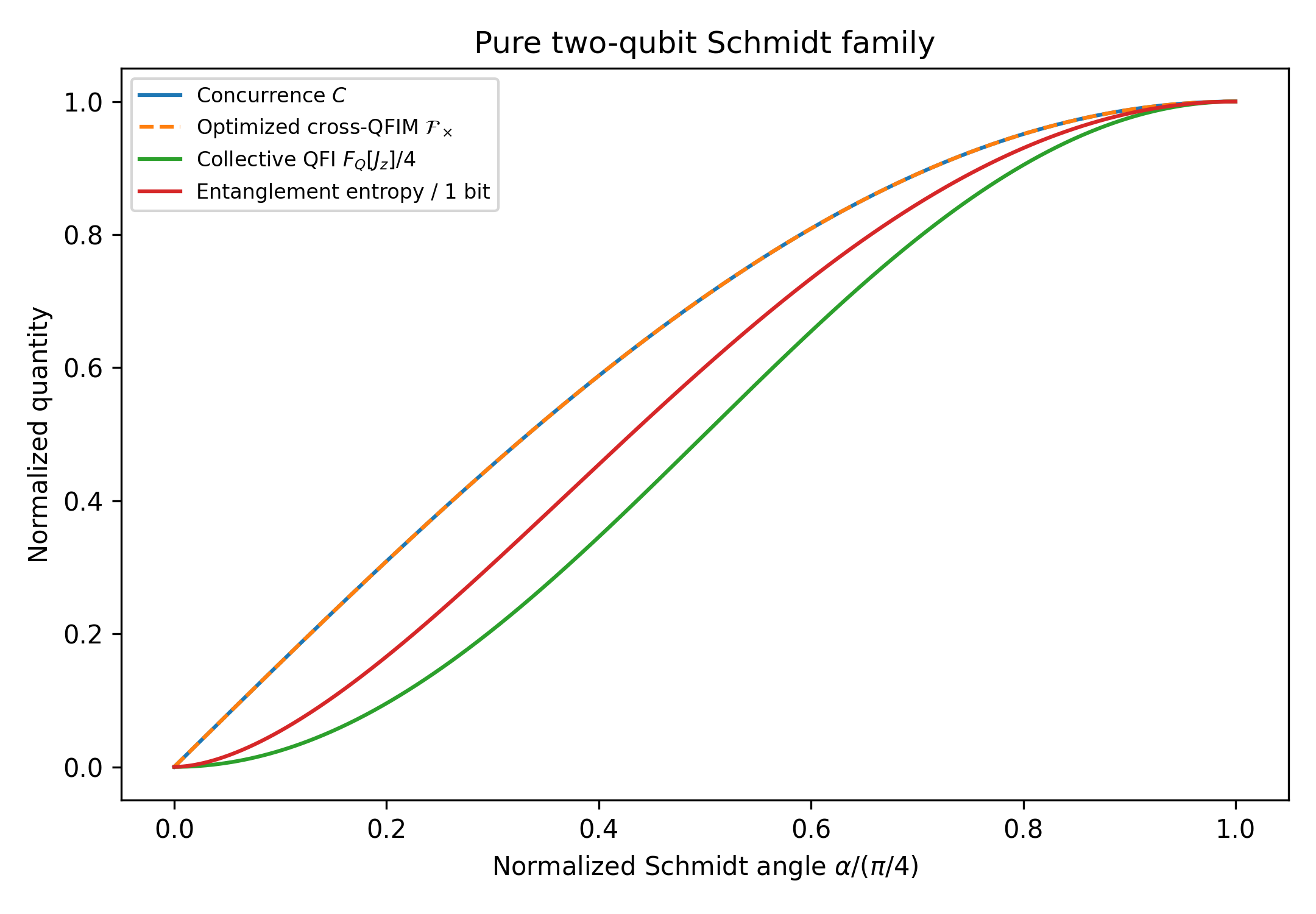}
\caption{Pure two-qubit Schmidt family. The optimized cross-QFIM quantity coincides with concurrence, $\mathcal F_\times=C$. The collective QFI normalized by four follows $C^2$, while the entanglement entropy is a different monotonic function of the Schmidt angle.}
\label{fig:schmidt}
\end{figure}

\section{Mixed states: raw cross-QFIM and convex-roof extension}
To test whether the cross-QFIM quantity remains an entanglement measure for mixed states, consider the Werner family \citep{Werner1989}
\begin{equation}
\rho_W(p)=p|\Psi^-\rangle\langle\Psi^-|+\frac{1-p}{4}I,
\qquad 0\le p\le1,
\end{equation}
where
\begin{equation}
|\Psi^-\rangle=\frac{|01\rangle-|10\rangle}{\sqrt2}.
\end{equation}
The concurrence follows from the two-qubit formula of \citet{Wootters1998}:
\begin{equation}
C_W(p)=\max\left[0,\frac{3p-1}{2}\right],
\end{equation}
so entanglement begins at $p>1/3$.

For mixed states the pure-state covariance identity is no longer sufficient. Using the symmetric-logarithmic-derivative QFIM, the unitary-encoding matrix elements may be written in the spectral basis $\rho=\sum_i\lambda_i|i\rangle\langle i|$ as \citep{BraunsteinCaves1994,PetzGhinea2011}
\begin{equation}
F_{\mu\nu}=2\sum_{i,j;\,\lambda_i+\lambda_j>0}
\frac{(\lambda_i-\lambda_j)^2}{\lambda_i+\lambda_j}
\,\mathrm{Re}\!\left[\langle i|G_\mu|j\rangle\langle j|G_\nu|i\rangle\right].
\end{equation}
For matched local generators $G_A=\sigma_x^A/2$ and $G_B=\sigma_x^B/2$, this expression gives
\begin{equation}
F_{AA}=F_{BB}=\frac{2p^2}{1+p},\qquad
F_{AB}=-\frac{2p^2}{1+p}.
\end{equation}
Because the Werner state is invariant under common local rotations $U\otimes U$, the magnitude is independent of the chosen matched Pauli axis and is the maximum over normalized local directions. Hence
\begin{equation}
\boxed{\mathcal F_\times^{(W)}=\frac{2p^2}{1+p}.}
\end{equation}
This quantity is nonzero for every $p>0$, including the separable interval $0<p\le1/3$. Therefore the raw cross-QFIM element is not a universal mixed-state entanglement measure. It quantifies parameter-space correlation more broadly. In a mixed state that correlation can contain classical correlations as well as nonclassical correlations that do not amount to entanglement, such as discord-type correlations; entanglement is a stricter property of the density operator.

\subsection{Convex-roof extension and exact mixed-state concurrence}
The pure-state identity in Proposition~1 has a direct mixed-state extension if the optimization is performed at the level of pure-state decompositions rather than on the QFIM of the mixed density operator itself. Define
\begin{equation}
\mathcal F_\times^{\rm cr}(\rho)=
\inf_{\{q_k,|\psi_k\rangle\}}
\sum_k q_k\,\mathcal F_\times(|\psi_k\rangle),
\qquad
\rho=\sum_k q_k|\psi_k\rangle\langle\psi_k|.
\end{equation}
The infimum is over all pure-state ensembles realizing $\rho$, as in standard ensemble-decomposition and convex-roof constructions \citep{Hughston1993,Toth2013}.

\begin{corollary}
For every two-qubit density operator $\rho$,
\begin{equation}
\boxed{\mathcal F_\times^{\rm cr}(\rho)=C(\rho),}
\end{equation}
where $C(\rho)$ is Wootters concurrence.
\end{corollary}
\begin{proof}
Wootters concurrence is the convex-roof extension of the pure-state concurrence \citep{Wootters1998}. Proposition~1 gives $\mathcal F_\times(|\psi_k\rangle)=C(|\psi_k\rangle)$ for every pure two-qubit component. Substituting this identity term by term into the defining infimum makes the two convex-roof optimizations identical.
\end{proof}

The distinction between $\mathcal F_\times^{\rm cr}(\rho)$ and the raw cross-QFIM evaluated directly on $\rho$ is essential. The former is an entanglement monotone because it is exactly concurrence for two qubits; the latter measures parameter-space coupling of the mixed state and can remain nonzero on separable states.

For the Werner family,
\begin{equation}
\mathcal F_\times^{\rm cr}(\rho_W)=C_W(p)=\max\left[0,\frac{3p-1}{2}\right].
\end{equation}
It is useful to quantify the excess of the raw mixed-state geometry above its convex-roof entanglement content,
\begin{equation}
\Delta_F(p)=\mathcal F_\times^{(W)}-\mathcal F_\times^{\rm cr}(\rho_W).
\end{equation}
For Werner states this difference is analytic:
\begin{equation}
\boxed{
\Delta_F(p)=
\begin{cases}
\dfrac{2p^2}{1+p}, & 0\le p\le\dfrac13,\\[2mm]
\dfrac{(1-p)^2}{2(1+p)}, & \dfrac13<p\le1.
\end{cases}}
\end{equation}
The quantity $\Delta_F$ should not be interpreted as a universal measure of classical correlation or quantum discord. It is simply the Werner-family difference between the raw mixed-state cross-QFIM coupling and the convex-roof part that is exactly attributable to two-qubit entanglement. It vanishes at both $p=0$ and $p=1$ and is positive for intermediate mixed Werner states.

Within the Werner family, the monotonic dependence on $p$ also allows useful family-specific landmarks. At the separability boundary $p=1/3$, the convex-roof quantity is still zero while the raw cross-QFIM has already reached
\begin{equation}
\mathcal F_\times^{(W)}=\frac16.
\end{equation}
For $p>1/3$, $\mathcal F_\times^{\rm cr}=C_W$ becomes positive.
Using the two-qubit CHSH criterion of the Horodeckis \citep{Horodecki1995}, the maximal CHSH value for this Werner family is
\begin{equation}
B_{\max}=2\sqrt2\,p,
\end{equation}
so CHSH violation begins at $p>1/\sqrt2$. The corresponding cross-QFIM value is
\begin{equation}
\mathcal F_\times^{(W)}=2-\sqrt2\approx0.5858.
\end{equation}
The family therefore separates into three regimes relevant to the present comparison: separable states with nonzero cross-QFIM, entangled yet CHSH-local states, and CHSH-nonlocal states. The first regime should not be labelled purely classical: separable Werner states can still carry nonclassical correlations in measures other than entanglement.

\begin{figure}[t]
\centering
\includegraphics[width=0.88\linewidth]{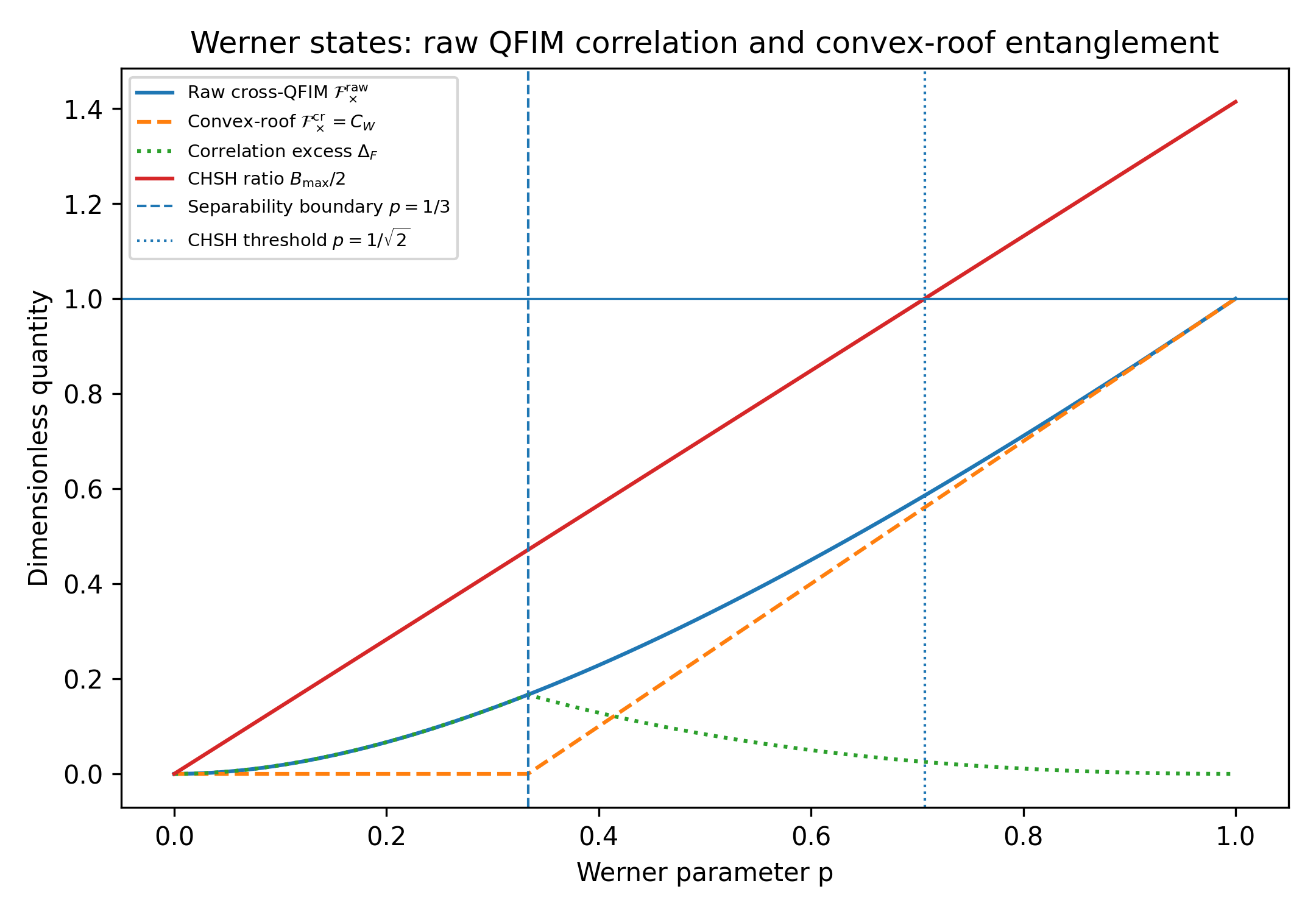}
\caption{Werner states. The raw optimized cross-QFIM $\mathcal F_\times^{\rm raw}=2p^2/(1+p)$ is nonzero for every $p>0$, whereas the convex-roof extension $\mathcal F_\times^{\rm cr}$ coincides exactly with the Werner concurrence. Their difference $\Delta_F$ quantifies the family-specific excess of raw parameter-space coupling over the entanglement contribution. CHSH violation begins only at $p>1/\sqrt2$.}
\label{fig:werner}
\end{figure}

\section{Decoherence and Fisher-information redistribution}
Consider a symmetric two-path state entangled with a which-way marker,
\begin{equation}
|\Psi\rangle=\frac{|1\rangle|M_1\rangle+e^{i\phi}|2\rangle|M_2\rangle}{\sqrt2},
\end{equation}
with marker overlap $\gamma=\langle M_1|M_2\rangle$. Writing $\gamma=|\gamma|e^{i\vartheta}$, its phase can be absorbed without loss of generality into the definition of the controllable relative phase, $\phi\mapsto\phi+\vartheta$. We therefore take the effective marker overlap to be real and nonnegative in the output probabilities below. After tracing over the marker, the reduced path state contains off-diagonal terms proportional to $\gamma$, and the interference visibility is
\begin{equation}
V=|\gamma|.
\end{equation}
For an ideal symmetric pure marker model, path distinguishability is
\begin{equation}
D=\sqrt{1-|\gamma|^2},
\end{equation}
saturating the duality relation $D^2+V^2=1$ \citep{Englert1996}.

For a two-output interference measurement,
\begin{equation}
P_\pm(\phi)=\frac12\left[1\pm|\gamma|\cos\phi\right],
\end{equation}
the classical Fisher information for estimating the phase is
\begin{equation}
I_\phi=\sum_{\pm}\frac{(\partial_\phi P_\pm)^2}{P_\pm}
=\frac{|\gamma|^2\sin^2\phi}{1-|\gamma|^2\cos^2\phi}.
\end{equation}
At the optimal quadrature operating point $\phi=\pi/2$,
\begin{equation}
\boxed{I_\phi^{\rm opt}=|\gamma|^2=V^2.}
\end{equation}
Thus the phase information accessible from subsystem interference vanishes as which-way states become orthogonal. Globally, however, the system--marker state remains correlated. It is therefore more precise to say that the information available to a local interference measurement has been redistributed into correlations with the marker. Related descriptions of the distribution and transfer of Fisher information in composite quantum systems have been developed independently \citep{Lu2012}; the present two-path calculation provides a particularly transparent instance of that idea in the context of decoherence \citep{Zurek2003,Joos2003,Schlosshauer2005}.

\begin{figure}[t]
\centering
\includegraphics[width=0.88\linewidth]{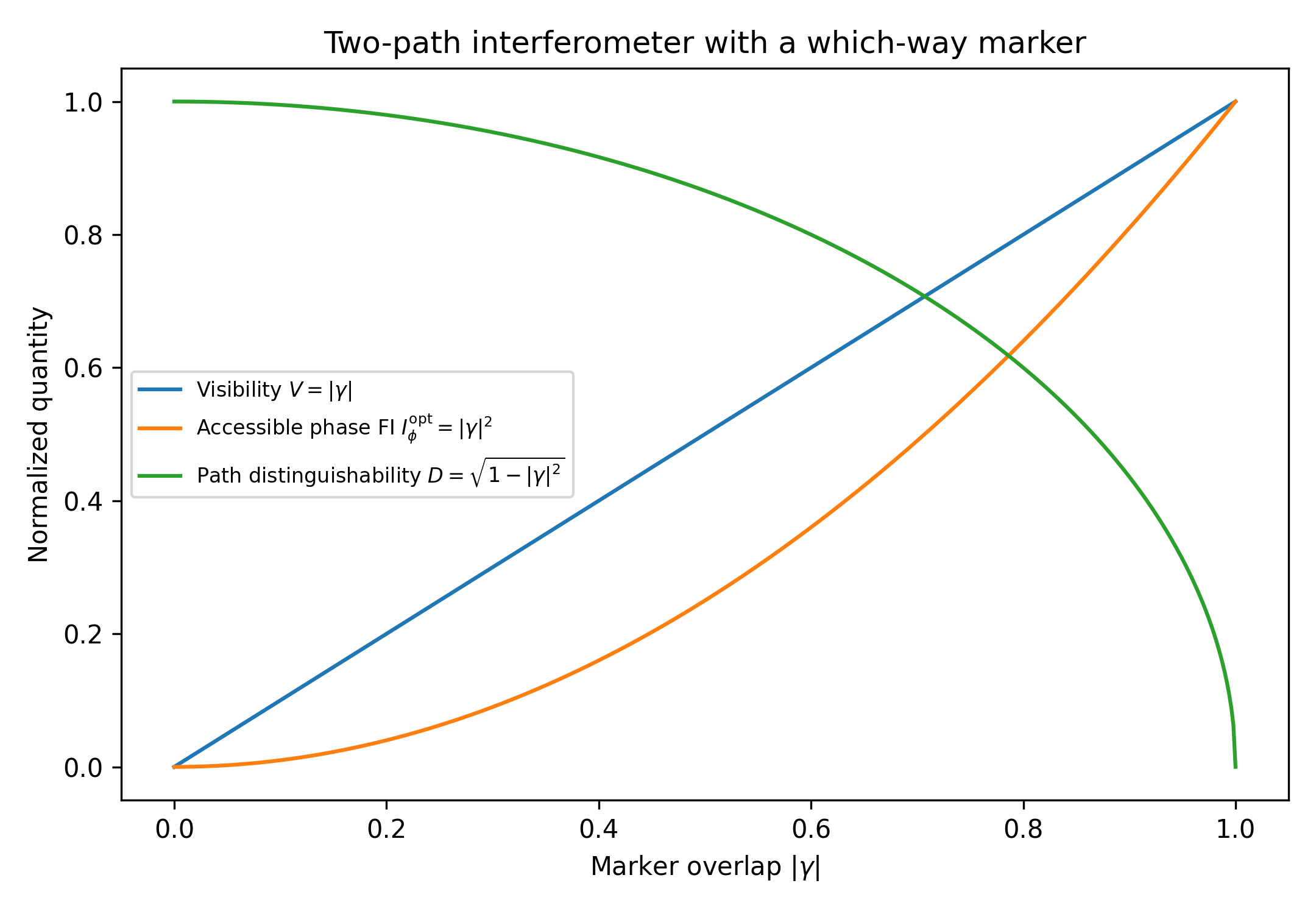}
\caption{Two-path interferometer with which-way marker. The accessible phase Fisher information at quadrature is $I_\phi^{\rm opt}=|\gamma|^2=V^2$, while path distinguishability increases as the marker states become more orthogonal.}
\label{fig:decoherence}
\end{figure}

\section{Analysis of results}
Taken together, the examples reveal a hierarchy rather than a single universal Fisher criterion. The diagonal QFI elements describe local sensitivity and are not entanglement indicators by themselves. The off-diagonal QFI element describes coupling between local statistical directions. For pure product states it vanishes for all strictly local generators; for pure entangled states, suitable local generators generate nonzero cross terms, and optimization yields the concurrence for two qubits.

The mixed-state analysis sharpens this result rather than merely breaking it. Direct evaluation of the QFIM on a mixed density operator produces a raw cross-parameter coupling that can be nonzero for separable states. By contrast, the convex-roof extension of the optimized pure-state quantity is exactly the two-qubit concurrence. Thus the two constructions answer different questions: $\mathcal F_\times^{\rm raw}$ probes correlation in the parameterized mixed-state geometry, whereas $\mathcal F_\times^{\rm cr}$ isolates entanglement because the ensemble optimization removes contributions that can be represented without irreducible two-qubit entanglement. This distinction is consistent with the broader separation between total, classical, discord-like quantum, and entanglement correlations in mixed states \citep{Modi2012,Adesso2016}. The Werner difference $\Delta_F$ is consequently useful as a controlled example of correlation beyond the convex-roof entanglement contribution, but it is not promoted here to a universal discord or classical-correlation measure.

The Werner family also makes visible the strict distinction between entanglement and Bell nonlocality. Bell's theorem and the CHSH construction rule out broad classes of local hidden-variable accounts of the observed correlations \citep{Bell1964,CHSH1969}, but entanglement is already present well before the Werner state violates CHSH. The thresholds $p=1/3$ and $p=1/\sqrt2$ therefore mark different resources, consistent with the broader hierarchy of entanglement, steering, and Bell nonlocality \citep{HorodeckiReview2009,Wiseman2007}.

The interferometric example supplies a complementary operational picture. As the marker overlap decreases, fringe visibility and locally accessible phase Fisher information decrease together, while path distinguishability increases. The calculation therefore tracks where a specific piece of phase sensitivity becomes inaccessible to local detection as correlations with the marker grow; it should not be read as a general conservation law for Fisher information.

\begin{table}[t]
\centering
\small
\caption{What the principal Fisher quantities diagnose in the models analyzed here. The Werner thresholds are family-specific and are not universal entanglement criteria.}
\label{tab:summary}
\begin{tabular}{>{\raggedright\arraybackslash}p{0.27\linewidth}>{\raggedright\arraybackslash}p{0.27\linewidth}>{\raggedright\arraybackslash}p{0.36\linewidth}}
\toprule
Quantity & Exact result & Interpretation / limitation \\ \midrule
$E_F=(\hbar^2/8m)I_F$ & Generates the quantum-potential term in the ensemble variational equations & Probability-gradient contribution to kinetic energy; requires global phase quantization for full Schr\"odinger equivalence \\ 
$\mathcal F_\times$ (pure two qubits) & $\mathcal F_\times=C$ & Exact optimized cross-QFIM representation of pure-state concurrence \\ 
$\mathcal F_\times^{\rm cr}$ (two qubits) & $\mathcal F_\times^{\rm cr}(\rho)=C(\rho)$ & Convex-roof extension; exact two-qubit entanglement monotone \\ 
$\mathcal F_\times^{(W)}$ & $2p^2/(1+p)$ & Raw mixed-state cross-QFIM coupling; nonzero also for separable Werner states \\ 
$\Delta_F$ (Werner) & $\mathcal F_\times^{(W)}-C_W$ & Family-specific excess of raw cross-QFIM coupling over convex-roof entanglement \\ 
$B_{\max}$ (Werner) & $2\sqrt2\,p$ & CHSH nonlocality only for $p>1/\sqrt2$ \\ 
$I_\phi^{\rm opt}$ & $|\gamma|^2=V^2$ & Locally accessible phase information in the two-path marker model \\ \bottomrule
\end{tabular}
\end{table}

\section{Discussion}
The calculations bring together several uses of information geometry that are usually discussed in different contexts, while making clear that they are not the same object. Fisher information can enter a variational reconstruction of Schr\"odinger dynamics \citep{Reginatto1998,HallReginatto2002}; QFI measures the distinguishability of nearby parameterized quantum states \citep{BraunsteinCaves1994,PetzGhinea2011}; and QFI or covariance-based bounds can certify entanglement and metrological usefulness \citep{Hyllus2012,Toth2012,Gittsovich2010,Pezze2018}. The common feature is sensitivity to structure, but the operational meaning of that sensitivity depends on which space and which generator are being considered.

The convex-roof result clarifies the strongest connection established here. The raw mixed-state QFIM and the convex-roof extension are not competing definitions of the same quantity. The former is computed directly from a density operator and its parameterized tangent directions; the latter minimizes the optimized pure-state cross-QFIM over all pure-state ensembles. For two qubits, that minimization reproduces concurrence exactly. The distinction explains why the Werner raw cross-QFIM can remain nonzero in the separable region without contradicting the entanglement interpretation of the convex-roof quantity.

A useful physical reading of these results is that different quantum structures can be described through different statistical geometries without assuming that those geometries are identical. In the Madelung representation, $\psi=\sqrt\rho\,e^{iS/\hbar}$, the amplitude determines the configuration-space probability density and the phase determines flow. The classical Fisher functional contributes to the kinetic-energy decomposition through gradients of that density. By contrast, the QFIM lives on a parameterized quantum-state manifold, where its off-diagonal entries quantify coupling between chosen local parameter directions. Keeping this distinction explicit avoids attributing a single ontological meaning to every appearance of the word ``Fisher.''

This language is also compatible with the standard no-signalling interpretation of entanglement. For an entangled state, the fundamental object is the global vector or density operator in $\mathcal H_A\otimes\mathcal H_B$. The two subsystems do not possess independently complete state geometries. Conditioning on a measurement outcome changes the description appropriate to one part of the correlated global state, but this does not imply signal transmission faster than light. The framework therefore accepts Bell-type nonseparability rather than attempting to restore a local hidden-variable model.

At the interpretative level, the same language separates two ingredients that are sometimes conflated in discussions of measurement. System-apparatus correlations are created physically by interaction before an observer reads a result. One may therefore separate physical correlation formation from statistical conditioning. Likewise, nothing in the calculations selects a many-worlds ontology or rules it out. A coherent superposition can be treated operationally as one quantum state with interfering alternatives, while Everettian accounts assign additional interpretative meaning to the same unitary structure. The Fisher construction developed here does not decide between those readings; its contribution is the quantitative comparison of the underlying dynamical and correlation geometries.

\section{Generalized Fisher dynamics and falsifiability}
As formulated so far, the framework is empirically equivalent to standard quantum mechanics. Distinct predictions arise only after the information functional itself is modified. One possible extension is
\begin{equation}
I_{\mathbf G}=\int \rho\,(\nabla\ln\rho)^T\mathbf G(\nabla\ln\rho)\,d\mathbf{x},
\end{equation}
where $\mathbf G$ is positive definite. For the single-particle Cartesian model used above, the standard isotropic theory corresponds to $\mathbf G=I$. A nontrivial spatially dependent or physically anisotropic metric would generate a modified information potential $Q_{\mathbf G}[\rho]$ and could alter wave-packet spreading, tunnelling, interference phases, dispersion, or entanglement dynamics. A constant positive matrix, by contrast, may partly amount to a change of coordinates or an effective mass tensor and must therefore be interpreted with care.

Generalized exact-uncertainty constructions already show that modifying the underlying information/uncertainty structure can lead to nonlinear Schr\"odinger-type equations \citep{Rudnicki2016}. The more precise theoretical question is which additional assumptions are needed to select the standard functional. Linearity, unitarity, Galilean covariance, separability of independent systems, probability conservation, and no-signalling are natural candidates, but uniqueness is not established here. If these constraints single out the usual Fisher term, they would provide a stronger reconstruction of Schr\"odinger dynamics; if not, the surviving alternatives would define controlled extensions that could in principle be bounded experimentally.

\section{Analytical and numerical reproducibility}
All three plotted results are generated from closed-form expressions derived in the text. Figure~\ref{fig:schmidt} uses $C=\sin2\alpha$, $\mathcal F_\times=C$, $F_Q[J_z]/4=C^2$, and the binary entropy of the Schmidt eigenvalues. Figure~\ref{fig:werner} uses the raw mixed-state cross-QFIM $2p^2/(1+p)$, the convex-roof result $\mathcal F_\times^{\rm cr}=C_W$, the analytic difference $\Delta_F$, and the Horodecki CHSH maximum $2\sqrt2 p$. Figure~\ref{fig:decoherence} uses $V=|\gamma|$, $D=\sqrt{1-|\gamma|^2}$ for the ideal pure marker model, and $I_\phi^{\rm opt}=|\gamma|^2$. No fitting parameters or external datasets enter the figures. The mixed-state QFI calculation follows the standard spectral representation of the symmetric logarithmic derivative metric \citep{Helstrom1976,BraunsteinCaves1994,Paris2009}.

\section{Limitations}
Several limitations are essential. First, reconstructing local Schr\"odinger dynamics from Fisher information does not by itself prove that Fisher information is ontologically fundamental; it establishes a consistent reconstruction. Global equivalence also requires the appropriate quantum circulation/single-valuedness conditions, so the Wallstrom issue cannot be ignored \citep{Wallstrom1994}. Second, the raw cross-QFIM evaluated directly on a mixed density operator is not an entanglement monotone; the exact mixed-state identity requires the convex-roof extension and is established here only for two qubits, where concurrence supplies the relevant closed construction. Third, the Werner excess $\Delta_F$ is a family-specific diagnostic and is not a universal measure of discord or classical correlation. Fourth, decoherence suppresses interference and selects stable classical structures but does not, by itself, explain why one unique macroscopic record is experienced in each individual run \citep{Zurek2003}. Finally, the interpretation remains empirically equivalent to standard quantum mechanics unless the information functional itself is modified.

\section{Conclusions}
The results support a common information-geometric description of several quantum structures while also showing which operations are required to compare them consistently. In the ensemble variational formulation, the Fisher-information term
\begin{equation}
E_F=\frac{\hbar^2}{8m}I_F
\end{equation}
generates the quantum potential and converts Hamilton--Jacobi ensemble dynamics into the local Schr\"odinger dynamics. The equivalence remains subject to the usual global phase-quantization condition.

For pure two-qubit states, the optimized cross component of the local QFI matrix satisfies
\begin{equation}
\mathcal F_\times=C,
\end{equation}
while a collective phase generator gives
\begin{equation}
F_Q[J_z]=4C^2.
\end{equation}
These relations show both the usefulness and the generator dependence of QFI-based correlation geometry.

The main mixed-state result is that the pure-state identity admits an exact convex-roof extension:
\begin{equation}
\boxed{\mathcal F_\times^{\rm cr}(\rho)=C(\rho)}
\end{equation}
for every two-qubit density operator. This equality must be distinguished from the raw QFIM of the mixed state. For Werner states,
\begin{equation}
\mathcal F_\times^{(W)}=\frac{2p^2}{1+p},
\end{equation}
which remains nonzero in part of the separable region, whereas the convex-roof quantity is
\begin{equation}
\mathcal F_\times^{\rm cr}(\rho_W)=\max\left[0,\frac{3p-1}{2}\right].
\end{equation}
Their analytic difference $\Delta_F$ quantifies, within this family, the portion of the raw cross-parameter coupling that is not retained by the entanglement convex roof. The CHSH threshold remains distinct, emphasizing that parameter-space correlation, entanglement, and Bell nonlocality are different resources.

In the marked two-path interferometer, the optimal locally accessible phase Fisher information obeys
\begin{equation}
I_\phi^{\rm opt}=V^2,
\end{equation}
which provides a compact operational description of the loss of local phase sensitivity as which-way information becomes available through correlations with the marker.

Taken together, these results support a narrower conclusion than the claim that information alone ``solves'' quantum mechanics. The spatial Fisher functional and the QFIM are distinct mathematical objects, but both encode sensitivity to structure in their respective spaces. For two-qubit entanglement, the optimized pure-state cross-QFIM and its convex roof reproduce concurrence exactly; for mixed-state QFIM geometry, the Werner example shows why raw correlations must still be separated from entanglement. The framework therefore supplies a quantitative bridge between variational quantum dynamics, bipartite correlation geometry, and decoherence without conflating their operational meanings.

Two extensions now become particularly natural. First, the convex-roof construction should be tested beyond two qubits and compared with established higher-dimensional and multipartite resource measures. Second, on the dynamical side, one can ask whether separability of independent systems, Galilean covariance, linearity, probability conservation, and no-signalling uniquely select the standard Fisher functional or permit experimentally constrained generalizations.

\appendix
\section{Mixed-state QFI matrix for the Werner family}
For a density operator $\rho=\sum_k\lambda_k|k\rangle\langle k|$ undergoing a unitary multiparameter encoding generated by $G_\mu$, the symmetric-logarithmic-derivative QFI matrix can be written as
\begin{equation}
F_{\mu\nu}=2\sum_{k,l:\lambda_k+\lambda_l>0}
\frac{(\lambda_k-\lambda_l)^2}{\lambda_k+\lambda_l}
\,\mathrm{Re}\!\left[\langle k|G_\mu|l\rangle\langle l|G_\nu|k\rangle\right].
\end{equation}
The Werner state has singlet eigenvalue $\lambda_s=(1+3p)/4$ and a threefold-degenerate triplet eigenvalue $\lambda_t=(1-p)/4$. Local Pauli generators connect the singlet sector to appropriate triplet components, so only singlet--triplet matrix elements contribute. Since
\begin{equation}
\frac{(\lambda_s-\lambda_t)^2}{\lambda_s+\lambda_t}=\frac{p^2}{(1+p)/2}=\frac{2p^2}{1+p},
\end{equation}
the normalized generators $G_A=\sigma_x^A/2$ and $G_B=\sigma_x^B/2$ give
\begin{equation}
F_{AA}=F_{BB}=\frac{2p^2}{1+p},\qquad F_{AB}=-\frac{2p^2}{1+p},
\end{equation}
where the negative sign reflects the singlet anticorrelation. More generally, rotational symmetry gives $F_{AB}(\mathbf n,\mathbf m)=-[2p^2/(1+p)]\,\mathbf n\cdot\mathbf m$, so the maximum magnitude is attained for parallel or antiparallel local axes. This derivation also makes clear why the cross term survives in separable Werner states: the QFIM detects parameter-space correlations of the mixed density operator, not entanglement alone.

\section{Phase Fisher information in the marked two-path model}
For $P_\pm=(1\pm v\cos\phi)/2$, with $v=|\gamma|$, direct substitution in the classical Fisher definition gives
\begin{align}
I_\phi&=\sum_{s=\pm}\frac{(\partial_\phi P_s)^2}{P_s}\nonumber\\
&=\frac{v^2\sin^2\phi}{1-v^2\cos^2\phi}.
\end{align}
The maximum over the operating phase is reached at quadrature, $\phi=\pi/2$, yielding $I_\phi^{\rm opt}=v^2=V^2$. This equality refers to the specified binary output measurement; it should not be confused with a statement that every possible measurement on the global system-marker state has QFI equal to $V^2$.

\section*{Data availability}
No experimental datasets were used. The numerical curves shown in Figs.~\ref{fig:schmidt}--\ref{fig:decoherence} are generated directly from the analytical expressions given in the text.

\section*{Declaration of competing interest}
The authors declare no competing financial or non-financial interests.

\end{document}